\documentclass[10pt]{IEEEtran}
\usepackage{amsfonts,amsmath,amssymb}
\usepackage{epsfig}
\usepackage{psfrag}\usepackage{enumerate}
\newtheorem{thm}{Theorem}%[section]
\newtheorem{cor}{Corollary}
\newtheorem{lem}{Lemma}

\newtheorem{defn}{Definition}%[section]
\newtheorem{remark}{Remark}

\newcommand{\Mv}{\bold{M}}
\newcommand{\Dn}{D^{(n)}}
\newcommand{\En}{E^{(n)}}
\newcommand{\ev}{\bold{e}}
\newcommand{\Gt}{\tilde{G}}
\newcommand{\Pb}{\bar{P}}
\newcommand{\hb}{\bar{h}}
\newcommand{\hu}{\underline{h}}

\newcommand{\CE}{C_{MMSE}}
\newcommand{\jo}{e^{j\omega}}

\newcommand{\Shv}{\hat{\mathbf{S}}}

\newcommand{\tx}{\mbox}{}
\newcommand{\nn}{\nonumber}

\usepackage{xspace}
\usepackage{bbm}
\usepackage{mathrsfs}
\def\cov{\mathop{\rm Cov}\nolimits}%
{}

\newcommand{\Cc}{\mathcal{C}}

\newcommand{\Mc}{\mathcal{M}}
\newcommand{\Nc}{\mathcal{N}}

\newcommand{\Sc}{\mathcal{S}}

\newcommand{\Xc}{\mathcal{X}}
\newcommand{\Yc}{\mathcal{Y}}
\newcommand{\Zc}{\mathcal{Z}}

\newcommand{\Yi}{{Y_1,Y_2,\ldots,Y_{i-1}}}

\newcommand{\Xv}{{\bf X}}
\newcommand{\Yv}{{\bf Y}}

\newcommand{\Sv}{{\bf S}}

\newcommand{\xv}{{\bf x}}

\newcommand{\pen}{{P_e^{(n)}}}

\newcommand{\Kb}{{\bar{K}}}

\newcommand{\Mh}{{\hat{M}}}
\newcommand{\Sh}{{\hat{S}}}

\newcommand{\ft}{\tilde{f}}

\def\b{\beta}

\def\e{\epsilon}

\def\l{\lambda}

\DeclareMathOperator\E{\mathsf{E}}
\let\P\relax
\DeclareMathOperator\P{\mathsf{P}}
\newcommand{\N}{\mathrm{N}}
\newcommand{\U}{\mathrm{Unif}}

\def\textiid{i.i.d.\@\xspace}
\newcommand\iid{\ifmmode\text{ i.i.d. } \else \textiid \fi}

\newcommand{\Complex}{\mathbb{C}}
\newcommand{\Real}{\mathbb{R}}
\newcommand{\half}{\frac{1}{2}}

\author{Ehsan Ardestanizadeh, {\em Student Member, IEEE} and Massimo Franceschetti, {\em Member, IEEE}  
\thanks{E.~Ardestanizadeh, M.~Franceschetti are with the Department of Electrical and Computer Engineering,
University of California, San Diego, La Jolla, CA, 92093-0407,
email: ehsan@ucsd.edu, massimo@ece.ucsd.edu.
}}

\date{}

\title{Control-theoretic Approach to Communication with Feedback: Fundamental Limits and Code Design}
\begin{document}

\maketitle
\begin{abstract}
Feedback communication is studied from a control-theoretic perspective, mapping the communication problem to a control problem in which the control signal is received 
through the same noisy channel as in the communication 
problem, and the (nonlinear and time-varying) dynamics of the system determine a subclass of encoders available at the transmitter. 
The {\em MMSE capacity} is defined to be the supremum exponential decay rate of the mean square decoding error. This is upper bounded by the information-theoretic feedback capacity, which is the supremum of the achievable rates. A sufficient condition is provided under which the upper bound holds with equality. 
For the special class of stationary Gaussian channels, a simple application of Bode's integral formula shows that the feedback capacity, recently characterized by Kim, is equal to the maximum {\em instability} that can be tolerated by the controller under a given power constraint. Finally, the control mapping is generalized to the $N$-sender AWGN multiple access channel. It is shown that Kramer's code for this channel, which is known to be sum rate optimal in the class of generalized linear feedback codes, can be obtained by solving a linear quadratic Gaussian control problem. %These results generalize the previous ones by Elia. 
%Finally, the exponential decay rate of the minimum mean square error (MMSE) in the point-to-point communication problem is considered and it is shown that Bode integral--a classical result in control theory--provides an upper bound on the asymptotic exponent.
\end{abstract}

\section{Introduction}
Feedback loops are central to many engineering systems. Their study naturally falls at the intersection between communication and control theories. However, the information-theoretic approach and the control-theoretic one have often evolved in isolation, separated by almost philosophical differences. In this paper we attempt one step at  bridging the gap, showing how tools from both disciplines can be applied to study fundamental limits of feedback systems and to design efficient codes for communication in the presence of feedback.

%In this paper we study fundamental limitations and code designs for  feedback communication problems from an estimation and control-theoretic point of view. 

Consider the feedback communication problem over an arbitrary point-to-point channel depicted in Fig 1a. The encoder, which has access to the channel outputs causally and noiselessly, wishes to communicate a continuous message $M \in (0,1)$ to the decoder through $n$ channel uses. At the end of the transmissions, the decoder forms an estimate $\Mh$ based on the received channel outputs, and the {\em mean square error} (MSE) of the estimate $\Mh$ represents the performance metric of the communication.

We map this communication problem to the general ({\em nonlinear and time varying}) controlled dynamical system depicted in~Fig.~\ref{fig:con-com}b, in which the initial state of the system corresponds to the message $M$, and the control actions--received through the same noisy channel as in the communication problem--correspond to the transmitted signals by the encoder. In this representation, the set of controllers for a given system corresponds to a subclass of encoders for the communication problem. In fact, the system can be viewed as a filter which determines the {\em information pattern}~\cite{Witsenhausen2}, on which the transmitted signals (actions) by the encoder (controller) can depend upon. A similar mapping for the special case of {\em linear time-invariant} (LTI) systems and controllers was first presented in~\cite{Elia2004}.% Since the decoder does not affect the distribution of the random variable

We study three different channel models. First, we consider a general point-to-point channel. The {\em MSE exponent} is defined as the exponential decay rate of the MSE in block length $n$, and the (feedback) {\em minimum mean square error (MMSE) capacity} is defined as the supremum of all achievable MSE exponents with feedback. We show that the MMSE capacity is upper bounded by the information-theoretic (feedback) capacity, the supremum of all achievable rates with feedback. We also present a sufficient condition, under which the (information-theoretic) capacity coincides with the MMSE capacity. These results provide a step towards the understanding of the connection between estimation and information theory.

Second, we focus on the stationary Gaussian channel with feedback, the capacity of which was recently characterized by Kim~\cite{YH}. Applying the discrete extension of Bode's result~\cite{DisBode1} (cf. \cite{Bode, Freudenberg}), we observe that the capacity of the Gaussian channel under power constraint $P$ is equal to the maximum instability which can be tolerated by a linear controller with power at most $P$, acting over the same stationary Gaussian channel. This follows almost immediately from the previous results~\cite{Elia2004,YH} and provides a step towards the understanding of the connection between stabilizability over some noisy channel and the capacity of that channel.
%For the special case of autoregressive moving-average ARMA(1) noise, Kim~\cite{YH} showed that the feedback capacity is achieved by the Schalkwijk and Kailath scheme~\cite{SchKai,Sch} which has doubly exponential error exponent and satisfies our sufficient condition for the MMSE capacity being equal to the capacity. 

%However, the control problem becomes equivalent to the communication problem if the cost to be optimized is picked according to the decoder and the performance criterion in the communication problem, e.g. probability of error or average distortion. %Hence, given a decoder\footnote{One can pick the optimal decoder according to the given criterion.}, the encoder can be designed based on the corresponding control problem. 

%Second, we consider the point-to-point communication problem over a stationary Gaussian channel with feedback. 
%and the decoder is the minimum mean square error (MMSE) estimate of the message. 
Finally, we consider the $N$-sender additive white Gaussian noise (AWGN) multiple access channel (MAC) with feedback depicted in Fig.~\ref{fig:GMAC}a. We show that the linear code proposed by Kramer~\cite{KramerFeedback}, which is known to be optimal among the class of generalized linear feedback codes~\cite{Ehsan}, can be obtained as the optimal solution of a linear quadratic Gaussian (LQG) problem given by a linear time-invariant (LTI) system controlled over a point-to-point AWGN channel where the asymptotic cost is the average power of the controller. These results provide a step towards the understanding of how control tools can be used to design codes for communication.

We now wish to place our results in the context of the related literature. The results on the $N$-sender AWGN-MAC  generalize previous ones of Elia~\cite{Elia2004}, who recovers Ozarow's code~\cite{Ozarow}--a special case of Kramer's code--using control theory. Our approach is different from Elia's in both
the model and the analysis. Our reduction is to a control
problem over a single point-to-point channel for any $N > 2$
number of senders, and our analysis is based on the theory
of LQG control. In contrast, in~\cite{Elia2004} the communication is limited
to $2$-sender MAC, which is mapped to a control problem
also over a 2-sender MAC, and the analysis is based on the
technique of Youla parameterization.

%%%%%%%%%%%%%%%%%%
The connection between the MMSE and capacity  has been investigated extensively and from different perspectives in the literature.  For example, in a classic paper Duncan~\cite{Duncan}  expresses the mutual information between a continuous random process and its noisy version corrupted by  white noise, in terms of the causal MMSE. More recently, Forney~\cite{Forney} explains the role of the MMSE  in the context of capacity achieving lattice codes over  AWGN channels. Guo et al.~\cite{Verdu} and Zakai~\cite{Zakai} showed that for a discrete random vector observed through an AWGN channel, the derivative of the  mutual information  between input and output sequences with respect to the signal-to-noise ratio (SNR), is half the (noncausal) MMSE. We point out that these authors study the average MMSE of a vector observed over a noisy channel without feedback as a function of the SNR. In contrast, we consider the estimation of a single random variable (message),  given the observation of a whole block of length $n$, and we look at the exponential decay rate of the MMSE with~$n$, at fixed SNR.  Of more relevance to us is the recent work of Liu and Elia~\cite{Elia2009}, who study linear codes over Gaussian channels obtained using a Kalman filter (KF) approach. For this class of codes, they show that the decay rate of the MMSE equals the mutual information between the message and the output sequence. In contrast, our results for the MMSE capacity are derived based on an information-theoretic approach and hold for all codes over general channels.

%We now wish to spend some additional words on the related literature. A complete survey of the literature is clearly beyond the scope of this paper, and we only mention a few works which were influential to ours. 
Additional works in the literature revealed connections between control theory and information theory.  Without attempting of being exhaustive, we distinguish between those who use information theory to study control systems and those who use control theory to study communication systems. Within the first group, Mitter and Newton~\cite{Mitter03,Mitter05} studied estimation and filtering in terms of information and entropy flows.   In the last decade, Bode-like fundamental limitations in controlled systems have been analyzed with success from an information-theoretic perspective~\cite{iglesias, iglesias2, Nuno07, Nuno08, Yu}. In this context, we point out that our Lemma~\ref{generalupp}, when it is specialized to the case of additive channels, provides an alternative proof of Theorem 4.2 in~\cite{Nuno08}. 

Within the second group, Elia~\cite{Elia2004} was the first to map linear codes for  additive white Gaussian noise channels to an LTI system controlled over an AWGN channel. Subsequently, Wu et al.~\cite{Sriram} studied the Gaussian interference channel in terms of estimation and control. Tatikonda and Mitter~\cite{Tatikonda} used a Markov  decision problem (MDP) formulation to study the capacity of Markov channels with feedback, and recently Coleman~\cite{Coleman} considered the design of the feedback encoder from a stochastic control perspective. Finally,  we also refer the reader to the work in \cite{MitterSurvey}, which gives an historical perspective and contains selected additional references.

The rest of the paper is organized as follows. Section~\ref{mapping} presents the definitions and the mapping between the feedback communication and the control problem. Section~\ref{MSErate} provides the upper bound on the MMSE capacity for a general point-to-point channel. The point-to-point stationary Gaussian channels and the AWGN multiple access channel are considered in Section~\ref{SGC} and Section~\ref{MAC}, respectively. Finally, Section~\ref{con} concludes the paper. 

Notation: A random variable is denoted by an upper case letter~(e.g. $X,Y,Z$) and its realization is denoted by a lower case letter~(e.g. $x,y,z$). Similarly, a random column vector and its realization are denoted by bold face symbols~(e.g. $\Xv$ and $\xv$, respectively). Uppercase letters~(e.g. $A,B,C$) also denote matrices, which can be differentiated from a random variable based on the context. The $(i,j)$-th element of $A$ is denoted by $A_{ij}$, and notation $A^T$ and $A'$ denote the transpose and complex transpose of matrix $A$, respectively. We use the following short notation for covariance matrices $K_{\Xv\Yv}:=\E(\Xv\Yv')-\E(\Xv)\E(\Yv')$ and $K_{\Xv}:=K_{\Xv\Xv}$. %Finally, $\L(\cdot)$ denotes an arbitrary linear function.
%Calligraphic letters~(e.g. $\Ac,\Bc,\Cc$) denote discrete sets, and $X(\Ac)$ represents an ordered subset of random variables $X_1,\ldots, X_N$ defined by set $\Ac$, i.e., $X(\Ac):= \{X_{j}: j\in \Ac \}$, $\m{A} \subseteq \Sc:=\{1,\ldots, N\}$. Similarly $X_i(\Ac):= \{X_{ji}: j\in \Ac \}$ denotes a subset of $\{X_{ji}\}$, $j=1,\ldots, N, i=1,\ldots,n$. Finally, $\L(\cdot)$ denotes an arbitrary linear function.

\begin{figure*}[htbp]
   \centering
   \begin{picture}(300,220)(0,0)
               
      \put(2,0){   
   \psfrag{f}[b]{ \hspace{7 em}\small $S_{i}=g_i(S_{i-1},Y_{i-1})$}
\psfrag{s}[b]{ \hspace{6.5em} System}
\psfrag{i}[b]{\small  $S_0$}
\psfrag{y}[b]{\hspace{1.5em} \small $Y_i$}
\psfrag{u}[b]{\hspace{0em} \small $X_i$}
\psfrag{z}[b]{\hspace{.3em} \small  $Z_i$}
\psfrag{q}{\hspace{.5em} { Controller }}
\psfrag{c}[b]{\hspace{6.5em} \small $Y_i=h_i(X_i,Z_i)$}
\psfrag{ch}[b]{\hspace{7em} Channel}

\psfrag{e}[b]{ \hspace{6.5em} \small  $X_i=\pi_i(S_{i})$}

   \includegraphics[width=3.15in]{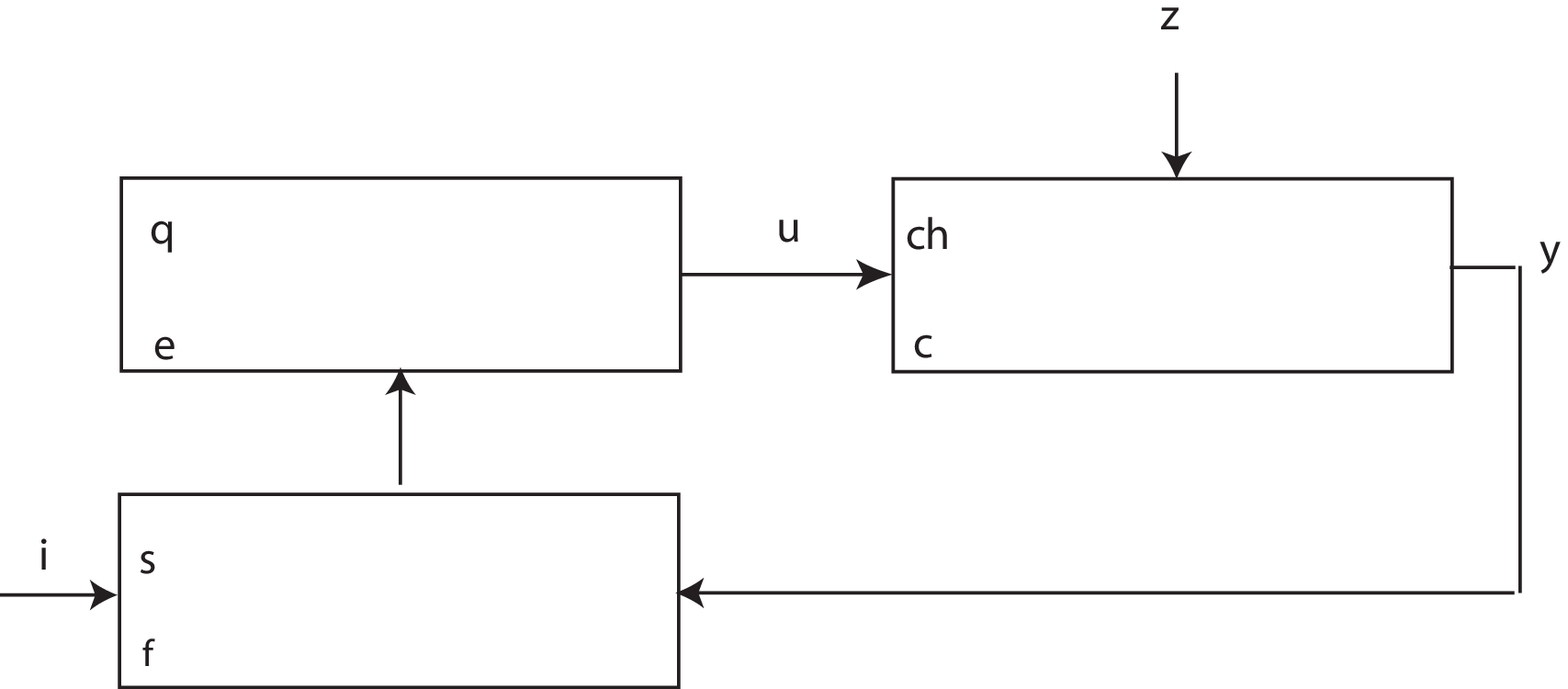} } 
   
  \put(-100,90){\scriptsize (a) Feedback communication  }  
   
   \psfrag{c}[b]{\hspace{6em}  \small $ Y_i=h_i(X_i,Z_i)$}
\psfrag{ch}[b]{\hspace{6.5em} Channel}
\psfrag{e}[b]{ \hspace{6em}  \small $X_i=f_i(M,Y^{i-1})$}
\psfrag{en}[b]{\hspace{6.8em}Encoder}
\psfrag{d}[b]{\hspace{4.5em}  \small$ \Mh=D(Y^{n})$}
\psfrag{de}[b]{\hspace{5em}Decoder}
\psfrag{m}[b]{\hspace{0em} \small $M$}
\psfrag{t}[b]{\hspace{1em}  \small $\Mh$}

\psfrag{u}[b]{\hspace{0em} \small $X_i$}
\psfrag{y}[b]{\hspace{0em}  \small$Y_i$}
\psfrag{z}[b]{\hspace{0em} {  \small $Z_i$}}
           
   \put(0,120){ 
           \includegraphics[width=4.5in]{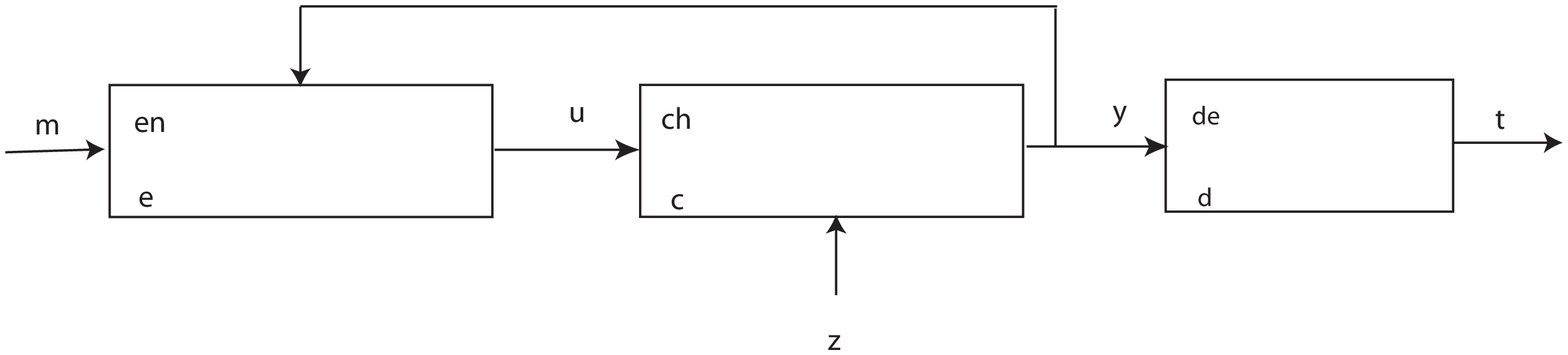}    }       
\put(-100,0){\scriptsize (b) Control over noisy channel  }  
                   
   \end{picture}   
   \vspace{.15in}
   \caption{An arbitrary encoder for the feedback communication over a stochastic channel depicted in (a) can be represented as a dynamical system controlled over the same channel depicted in (b).} 
  \label{fig:con-com}
\end{figure*}

%%%%%%%%%%%%%%%%%%%%%%%%%Section Definitions%%%%%%%%%%%%%%%%%%%

%%%%%%%%%%%%%%%%%%%%%%%%%Section Definitions%%%%%%%%%%%%%%%%%%%
\section{Definitions and Control Approach}\label{mapping}

Consider the communication problem depicted in Fig.~\ref{fig:con-com}a, where the sender wishes to convey a message $M \in \Mc:=(0,1)$ to the receiver through $n$ uses of a stochastic channel, 
\begin{align*}
Y_i=h_i(X_i,Z_i), \quad i=1,\ldots,n.
\end{align*}
where $X_i \in \Xc$ and $Y_i \in \Yc$ denote the input and output of the channel, respectively, and $Z_i \in \Zc$ denote the noise at time~$i$. The set of mappings 
\[h_i : \Xc \times \Zc \to \Yc, \quad  i=1,\ldots,n\]
and the distribution of the noise sequence $\{Z_i\}_{i=1}^n$ determine the channel. The noise process $\{Z_i\}$ is assumed to be independent of the message $M$. We assume that the output symbols are causally fed back to the sender and the transmitted symbol~$X_i$ at time $i$ can depend on both the previous
channel output sequence $Y^{i-1}:= \{Y_1, Y_2, . . . , Y_{i-1}\}$ and the
message $M$. 
\begin{defn}\label{ncodedef}
We define a (feedback) $n$-code as 
\begin{enumerate}
\item an encoder: a set of (stochastic)\footnote{For stochastic encoders we can write $X_i$ as a function of $(M,Y^{i-1},V_i)$, where $\{V_i\}$ is a random process independent of $M$ and $\{Z_i\}$.} encoding maps $f_i: \Mc \times \Yc^{i-1} \to \Yc$, also known to the receiver, such that for each %message $m \in (0,1)$ and
$i=1,\ldots,n$ 
 \begin{align}\label{transmitted}
 X_i=f_i(M,Y^{i-1})
 \end{align}
and
\item a decoder: a decoding map $\phi: \Yc^n \to \Mc$ which determines the estimate of the message $\Mh$  based on the received sequence $Y^n$, i.e., 
\begin{align}\label{decoderdef}
\Mh=\phi(Y^n).
\end{align}
\end{enumerate}
\end{defn}
We assume that the message $M \in (0,1)$ is a random variable uniformly distributed over the unit interval and does not depend on $n$. As the performance measure of the communication, we consider the MSE,%\footnote{The representation in this section also holds for other message sets and performance measures.} 
\begin{align}\label{costmmse}
D^{(n)}:=\E\Big((M-\phi(Y^n))^2\Big).
\end{align}
where the expectation is with respect to randomness of both the message and the channel.
%\[\E\Big((M-\phi(Y^n))^2\Big).\]
Note that the decoder does not affect the joint distribution of $(M,X^n,Y^n,Z^n)$ and simply estimates the message at the end of the block. Hence, without loss of generality, we pick the optimal decoder, namely, the MMSE estimator of the message given $Y^n$, and we call an encoder optimal if it minimizes $\Dn$. Let
\[\En:=-\frac{1}{2n}\log(\Dn)\]
be the exponential decay rate of the MSE with respect to~$n$. 
\begin{defn}
The {\em MSE exponent} $E$ is called achievable (with feedback) if there exists a sequence of $n$-codes such that
\begin{align}\label{Edef}
E \leq \liminf_{n \to \infty} \En.
\end{align}
\end{defn}
The {\em MMSE capacity} $\CE$ is the supremum of all achievable MSE exponents.
% The MMSE capacity can obviously be achieved by the optimal encoder for each~$n$, given the MMSE decoder. However, since it is an asymptotic quantity, it is not hard to see that it can also be achieved by other sequence of $n$-codes, where the encoders are not necessarily optimal for each~$n$.

The communication problem described above can be viewed as a control problem where the encoder wishes to control the knowledge of the receiver about the message. In his notable paper~\cite{Witsenhausen}, Witsenhausen wrote: ``When communication problems are considered as control problems (which they are), the information pattern is never classical since at least two stations, not having access to the same data, are always involved''. However, this is not the case for the feedback communication problem in Fig.~\ref{fig:con-com}a since the encoder (controller) has access also to the information available at the decoder via feedback. In fact, below we show that this feedback communication problem can be represented as a control problem in which the controller has the complete state information.

Consider the control problem in Fig.~\ref{fig:con-com}b where the state at time $i$ is 
\begin{align}\label{systemgen}
S_i=g_i(S_{i-1}, Y_{i-1}), \quad i=1,\ldots,n
\end{align}
with initial state 
\[S_0=M\]
and $Y_0=\emptyset$. We refer to the mappings $\{g_i\}_{i=1}^n$, 
\[g_i : \Sc_{i-1} \times \Yc \to  \Sc_{i},  \quad \quad i=1,\ldots,n \]
as the {\em system}. The controller, which observes the current state~$S_{i}$, picks an action (symbol) $X_i \in \Xc$, 
\begin{align}\label{transmitted2}
X_i=\pi_i(S_{i}), \quad i=1,\ldots,n 
\end{align}
according to a set of (stochastic) %\footnote{Stochastic mapping can be consider as $X_i=\tilde{\pi}_i(S_{i-1},Q_{i})$, where $\{Q_i\}$ is a random process independent of $\{Z_i\}$ and $S_0$.   } 
mappings
\[\pi_i : \Sc_{i} \to  \Xc , \quad i=1,\ldots,n.\]
We refer to the set $\{\pi_i\}_{i=1}^n$ as the {\em controller}. 

The communication problem in Fig.~\ref{fig:con-com}a can be represented as the control problem Fig.~\ref{fig:con-com}b as follows. Let the system $\{g_i\}_{i=1}^n$ be such that the state $S_i$ at time~$i$ is the collection of the initial state $S_0=M$ and the past observations $Y^{i-1}$, namely,
\begin{align}\label{stateally}
S_i=(M,Y^{i-1}) \quad i=1,\ldots, n.
\end{align}
Also, let the controller $\{\pi_i\}_{i=1}^n$ be picked according to the encoder $\{f_i\}_{i=1}^n$ in the communication problem such that 
\[ \pi_i(S_{i}) = f_i(M,Y^{i-1}), \quad \quad i=1,\ldots,n.\]
Then the joint distribution of all the random variables $(M,X^n,Y^n,Z^n)$ in the control problem is the same as that in the communication problem. %Hence, by picking state as in \eqref{stateally}, an arbitrary encoder $\{f_i\}_{i=1}^n$ can be viewed as a controller $\{\pi_i\}_{i=1}^n$.

To complete the representation, let $\Sh_0(Y^n)$ be the MMSE estimate of the initial state $S_0$ based on~$Y^n$, and the final cost be
\[c_n(S_n):=(S_0-\Sh_0(Y^n))^2.\]
We call a controller optimal if it minimizes the final expected cost
\begin{align*}
\E(c_n(S_n))&=\E\Big((M-\Mh(Y^n))^2\Big)\\
&=\Dn.
\end{align*}
Thus, the optimal controller represents the optimal encoder for the communication problem.

The system in~\eqref{stateally} is the most general system which can represent all the encoders for the communication system. However, if the system $\{g_i\}_{i=1}^n$ is more restricted such that the state $S_i$ is a filtered version of $(S_0,Y^{i-1})$, then the controller $\{\pi_i\}_{i=1}^n$ represents only a subclass of encoders $\{f_i\}_{i=1}^n$ where the transmitted symbol $X_i$ depends on $(M,Y^{i-1})$ only through $S_i$ (see \eqref{transmitted2}). In that case, we can view the system as a {\em filter} which determines the {\em information pattern} available at the controller (encoder). Below, we show a subclass of encoders for which the state $S_i$ does not include all the past output $Y^{i-1}$ as in~\eqref{stateally}, yet it contains all the optimal encoders. Let $F_{M|Y^{i-1}}(\cdot|Y^{i-1})$ be the conditional distribution of the message $M$ given the channel outputs $Y^{i-1}$.
\begin{lem}\label{infostate}
An optimal encoder which minimizes the MSE $\Dn$, can be found in the subclass of encoders which is determined by the system with state of the form $S_i=(M,F_{M|Y^{i-1}}(\cdot|Y^{i-1}) )$.
\end{lem}
\IEEEproof See Appendix~\ref{appleminfo}.
%\begin{remark}

This lemma, which is based on a well-known result in stochastic control, provides a sufficient information pattern such that the optimal feedback encoders for the communication over a general channel can be built upon. For example, in the special case of memoryless channels, Shayevitz and Feder~\cite{Oferarxiv} proposed an explicit encoder which uses the information pattern described in Lemma~\ref{infostate}, and showed that it is optimal in terms of the achievable rates.  
%%%%%%%%%%%%%%%%%%%%%%%%%%%%MSE exponent vs Rate

%%%%%%%%%%%%%%%%%%%%%%%%%%%%%%%%%%%%%%%%%%%%%%%%%%%
\section{MMSE Capacity}\label{MSErate}

In this section we present the relationship between the MMSE capacity and the information-theoretic capacity~\cite{Cover--Thomas2006}. 

First, we review the definition of the information-theoretic achievable rates, which is an asymptotic measure based on a sequence of message sets whose size depend on the block length~$n$. Consider the $n$-code in Definition~\ref{ncodedef} and let the message $M_n \sim \U(\Mc_n)$ be uniformly distributed over the set  $\Mc_n:=\{1,2,\ldots,2^{nR}\}$ such that 
 \begin{align}
 X_i=f_i(M_n,Y^{i-1}) , \ i=1,\ldots,n.
 \end{align}
Let probability of error be $\pen=\P(M_n \neq \Mh_n)$. A rate $R$ is called achievable (with feedback) if there exists a sequence of $n$-codes with message sets $\Mc_n$ for which $\pen \to 0$ as $n \to  \infty$. The (feedback) capacity $C$ is the supremum of all achievable rates. 

Recall that the MSE exponent $E$ is defined based on a message $M \in (0,1)$ which does not depend on the block length~$n$. The following lemma provides the connection between the achievable rates and MSE exponents.
\begin{lem}\label{thmMSE}
If the MSE exponent $E$ is achievable, then any rate $R<E$ is achievable. Conversely, if the rate $R$ is achievable and the probability of error satisfies $\lim_{n \to \infty} \frac{\pen}{2^{-2nR}} = 0$, then the MSE exponent $E=R$ is achievable.
\end{lem}
\IEEEproof See Appendix~\ref{appthmMSE}.

According to Lemma~\ref{thmMSE}, showing that an MSE exponent $E > R$ is achievable for a uniform message $M \in (0,1)$ is sufficient to show that rate $R$ is achievable.

Below, we provide a general upper bound on the MMSE capacity. Before proceeding, we need some information-theoretic definitions~\cite{Cover--Thomas2006}. Let $h(X)$ denote the differential entropy of a random variable $X$ and $I(X;Y)$ be the mutual information between random variables $X$ and $Y$. Moreover, let the {\em minimum entropy rate} of a general random process $\{X_i\}_{i=1}^\infty$ be defined as
\begin{align*}
%\hb(X):=&\limsup{n \to \infty} \frac{h(X^n)}{n} \\
 \hu(X):=&\liminf_{n \to \infty} \frac{h(X^n)}{n}.
 \end{align*}
For a stationary random process, the limit exists~\cite{Cover--Thomas2006} and the {\em entropy rate} is defined as
\begin{align*}
%\hb(X):=&\limsup{n \to \infty} \frac{h(X^n)}{n} \\
 \hb(X):=&\lim_{n \to \infty} \frac{h(X^n)}{n}.
 \end{align*}
 We call a channel invertible if $Z_i$ can be recovered from $(X_i,Y_i)$, namely, for every time $i > 0$ there exists a function $h_i^{-1}$ such that 
\[Z_i=h_i^{-1}(X_i,Y_i).\]

\begin{lem}\label{generalupp}
For a system controlled over an invertible channel, we have
\[I(S_0;Y^n) \leq h(Y^n)-h(Z^n) \]
where equality holds if and only if $X_i$ is a deterministic function of $(S_0,Y^{i-1})$ for $i=1,\ldots,n$.
\end{lem}
\IEEEproof See Appendix~\ref{appgeneralupp}.

\begin{remark}
Theorem 4.2 in~\cite{Nuno08} can be recovered from Lemma~\ref{generalupp} by considering the special case of additive channels $Y_i=X_i+Z_i$ which are invertible.
\end{remark}

\begin{thm}\label{MSEupp}
For an invertible channel, the MMSE capacity $\CE$ is upper bounded as
\[ \CE \leq C \leq \sup \ \hu(Y)-\hu(Z).\]
where the supremum is over all sequences of $n$-codes.
%where the maximization is over causal input distributions
%\[p(x_i|x^{i-1},y^{i-1}), \ i=1,\ldots, n.\]
\end{thm}
\begin{remark}
Theorem~\ref{MSEupp} holds without any assumption on the noise sequence or the controlled system. 
\end{remark}
\begin{IEEEproof}

The first inequality
\begin{align}\label{ine1}
\CE \leq C
\end{align}
can be shown using Lemma~\ref{thmMSE}. Suppose \eqref{ine1} does not hold and $ C < \CE $. Then, we can always find an achievable MSE exponent~$E^*$ such that $C < E^*  $, and subsequently we can pick a rate $R^*$ such that $C < R^* < E^*$. Hence, by Lemma~\ref{thmMSE} the rate $R^*$ is achievable, and this contradicts the definition of $C$. Therefore, the inequality~\eqref{ine1} must hold. 

For the second inequality, from Fano's inequality we have
\[R \leq \frac{1}{n} I(S_0,Y^n)+\epsilon_n .\]
where $\e_n \to 0$ as $n \to \infty$. Hence, for a given sequence of $n$-codes
\begin{align}\label{Clim}
R \leq \liminf_{n \to \infty} \frac{1}{n} I(S_0,Y^n).
\end{align}
From Lemma~\ref{generalupp} we have $h(Z^n)+I(S_0;Y^n) \leq h(Y^n)$. Taking the limit and rearranging terms
we get
\[\liminf_{n \to \infty} \frac{1}{n} I(S_0,Y^n) \leq \hu(Y)-\hu(Z).\]
Combining with \eqref{Clim} we have $R \leq \hu(Y)-\hu(Z)$, and taking the supremum over all sequences of $n$-codes completes the proof.
\end{IEEEproof}

Theorem~\ref{MSEupp} shows that in general, (feedback) MMSE capacity is upper bounded by the (feedback) capacity. Below, we provide a sufficient condition on the channel under which this bound is tight. 
\begin{cor}\label{sufficient}
For a given channel, we have
\[\CE=C\]
if any rate $R < C$ can be achieved such that
\begin{align}\label{errorexp}
\lim_{n \to \infty} \frac{\pen}{2^{-2nC}} \to 0.
\end{align}

\end{cor}
\begin{IEEEproof}
By assumption for any rate $R < C$, the condition on the decay rate of probability of error $\pen$ in Lemma~\ref{thmMSE} is satisfied and hence any MSE exponent $E=R < C$ is achievable. Therefore, we conclude that $\CE \geq C$. On the other hand, by Theorem~\ref{MSEupp} we have $\CE \leq C $, and hence $\CE=C$.
\end{IEEEproof}

In the next section we show that this equality holds for some special cases of Gaussian channels including the AWGN channel.  

%%%%%%%%%%%%%%%%%%%%%%%%%Stationary Gaussian %%%%%%%%%%%%%%%%%%%

%%%%%%%%%%%%%%%%%%%%%%%%%%%%%%%%%%%%%%%%%%%%%%%%%%%

\section{Point-to-Point Gaussian Channels}\label{SGC}
In this section, we turn our attention to additive Gaussian channels, 
\[Y_i=X_i+Z_i\]
where the noise sequence $\{Z_i\}$ is a (colored) stationary Gaussian process with power spectral density $S_Z(e^{j\omega})$. The transmitted symbols are assumed to satisfy the (block) power constraint $P$, i.e.,
\[\frac{1}{n} \sum_{i=1}^n \E(X_i^2) \leq P.\]

We define the subclass of {\em LTI encoders} to be the encoders which can be represented by an LTI system and a linear controller in Fig.~\ref{fig:con-com}b. In this case, the system and the controller combined can be represented as as a {\em causal} filter
\[F(z)=\frac{X(z)}{Y(z)}\]
where $X(z)$ and $Y(z)$ are the $Z$-transform of the input and the output sequence, respectively. We alternatively refer to $F(z)$ as the open loop transfer function. Let the instability $I$ be
\[I=\sum_{j=1}^m \log(|\beta_j|).\]
where $\beta_j$ are the $m$ unstable poles of the open loop transfer function $F(z)$. 
\begin{thm}\label{SG}
Under power constraint $P$, the capacity of the stationary Gaussian channel is
\[C=\sup \ I\]
where the supremum is over all causal filters $F(z)$ which satisfy the power constraint
\begin{align}\label{pc2}
\frac{1}{2\pi}\int_{-\pi}^{\pi} \left|\frac{F(e^{j\omega})}{1-F(\jo)}\right|^2 S_Z(e^{j\omega}) d\omega \leq P.
\end{align}
\end{thm}
\begin{remark}
Theorem~\ref{SG} links the two works and is  a simple consequence of~\cite{Elia2004,YH}. The Gaussian channel capacity can be represented as the limit of a sequence of optimization problems~\cite{Cover--Pombra}. Kim~\cite{YH} proved that the limit of this sequence can be replaced by a single optimization problem, maximizing the entropy rate of the output process over all stationary Gaussian input processes, which are a causal and linear filtered version of the noise process. On the other hand, Elia~\cite{Elia2004} expressed the asymptotic average directed information between the input and output process for a stationary linear scheme over a Gaussian channel in terms of the Bode integral. We connect the capacity, the supremum achievable rate of communication, to the supremum tolerable instability over  Gaussian channels. 
\end{remark}
\begin{IEEEproof}
The capacity of the stationary Gaussian channel under power constraint $P$ is given by \cite[Theorem 4.6]{YH} 
\begin{align}
C&=\sup_B \frac{1}{2\pi}\int_{-\pi}^{\pi} \half \log\left(|1+B(e^{j\omega})|^2\right)d\omega \label{Gcapa}
\end{align}
where the maximization is over all causal filters $B(e^{j\omega})$
%\[B(z)=\frac{X(z)}{Z(z)}\]
which satisfy 
\begin{align}\label{pc}
\frac{1}{2\pi}\int_{-\pi}^{\pi} |B(e^{j\omega})|^2 S_Z(e^{j\omega}) d\omega \leq P.
\end{align}
This characterization can also be viewed as follows~\cite{YH},
\begin{align}
C&=\sup_{\{X_i\}} \hb(Y) - \hb(Z)  \label{Gcapa2}
\end{align}
where the supremum is over all stationary Gaussian processes $\{X_i\}$ of the form $X_i=\sum_{k=1}^\infty b_k Z_{i-k}$ such that $\E(X_i^2) \leq P$. 
% the transfer function defined by a causal filter $B(z)=\sum_{k=1}^\infty b_k Z^{-k}$ as follows  
%\[X(z)=B(z)Z(z)\]

We show that the capacity expression~\eqref{Gcapa} can be rewritten in terms of the instability which can be tolerated under the power budget. Let the {\em sensitivity function}~\cite{Doyle} be the transfer function from the output to the noise, 
\begin{align}\label{sensitivitydef}
S(z)=\frac{Y(z)}{Z(z)}.
\end{align}
By discrete version of Bode's integral~\cite{DisBode1} we know that if the closed loop is stable then
\begin{align}\label{Bode}
\frac{1}{2\pi}\int_{-\pi}^{\pi} \log(|S(e^{j\omega})|)d\omega = \sum_{j=1}^m \log(|\beta_j|)=I.
\end{align}
%The equality implies that if the sensitivity to disturbance is suppressed at some frequencies, it has to increase at some other frequencies, and for that reason it is referred to as {\em water bed} effect. Kolmogrov's result states
%\[\frac{1}{2\pi}\int_{-\pi}^{\pi} \log(|S(e^{j\omega})|)d\omega = h(Y)-h(Z).\]
Let $S_Y(e^{j\omega})$ and $S_Z(e^{j\omega})$ be the power spectral density of the output and the noise sequence, respectively. From \eqref{sensitivitydef} we have
\begin{align}\label{powerspec}
|S(e^{j\omega})|^2=\frac{S_Y(e^{j\omega})}{S_Z(e^{j\omega})}
\end{align}
and plugging \eqref{powerspec} into \eqref{Bode} we get
\begin{align}\label{Bode2}
I=\frac{1}{2\pi}\int_{-\pi}^{\pi} \half \log\left(\frac{S_Y(e^{j\omega})}{S_Z(e^{j\omega})}\right)d\omega.
\end{align}

The entropy rate of a stationary Gaussian process $\{Y_i\}_{i=1}^\infty$ was shown by Kolmogorov~\cite{Cover--Thomas2006} to be
\begin{align}\label{entrate}
\hb(Y)= \frac{1}{2\pi} \int_{-\pi}^{\pi} \half \log(2\pi e S_Y(e^{j\omega}))d\omega.  
\end{align}
Plugging \eqref{entrate} into \eqref{Gcapa2} and comparing with~\eqref{Bode2} we have
\[C= \sup \ I\]
where the supremum is over all causal $B(z)$ such that \eqref{pc} is satisfied. On the other hand, $B(z)=\frac{X(z)}{Z(z)}$ can be written in terms of $F(z)=\frac{X(z)}{Y(z)}$ as follows,
\[B(z)=\frac{F(z)}{1-F(z)}.\]
Note that $F(z)$ being causal is equivalent to $B(z)$ being causal and vice-versa. Therefore, considering constraint \eqref{pc}, the supremum can be taken over all causal filters $F(z)$ such that
\[\frac{1}{2\pi}\int_{-\pi}^{\pi} \left|\frac{F(e^{j\omega})}{1-F(\jo)}\right|^2 S_Z(e^{j\omega}) d\omega \leq P.\]
\end{IEEEproof}

Theorem~\ref{SG} states that the capacity of the stationary Gaussian channel represents the maximum instability of an LTI system which can be stabilized over the Gaussian channel by a linear controller with power at most $P$. 

Next, we consider some special cases of the noise spectrum. 

%%%%%%%%%%%%%%%%%%%%%%%%%%%%%%%%%%%%%%%%%%%%%%%
\subsection{ARMA(1) Noise}
Let the noise process be first-order autoregressive moving-average ARMA(1), i.e.,
\begin{align}\label{ARMA}
S_Z(e^{j\omega}) = \frac{|1+\alpha e^{j\omega}|}{|1+\b e^{j\omega}|}
\end{align}
where $\alpha \in [-1,1]$ and $\b \in (-1,1)$. 
\begin{cor}\label{ARMAProp}
For the stationary Gaussian channel with noise process ARMA(1) given in~\eqref{ARMA}, and power constraint $P$, we have  
\begin{align}\label{equalcap}
\CE=C= \sup I
\end{align}
where the supremum is over all causal $B(z)$ such that the power constraint~\eqref{pc} is satisfied.
\end{cor}
\begin{IEEEproof}
To show the first equality in~\eqref{equalcap}, i.e., $\CE=C$, note that the feedback capacity of the Gaussian channel with ARMA(1) noise is achieved~\cite[Theorem 5.3]{YH} by the Schalkwijk-Kailath scheme with probability of error going to zero doubly exponentially fast. Therefore, the sufficient condition~\eqref{errorexp} provided in Section~\ref{MSErate} is satisfied and hence $\CE=C$. The second equality in~\eqref{equalcap} follows from Theorem~\ref{SG}. 
\end{IEEEproof}

\begin{remark}
Note that the AWGN channel, for which the noise spectrum is white, is a special case of ARMA(1) with $\alpha=\b=0$.
\end{remark}

%%%%%%%%%%%%%%%%%%%%%%%%%%%%%%%%%%%%%%%%%%%%%%%

\subsection{White Noise}

The following corollary considers the special case of AWGN channels.
% (cf.~\cite[Lemma 4.1]{YH}).
\begin{cor}\label{exAWGN}
For AWGN channels where the noise is i.i.d.\@ and Gaussian, $Z_i \sim \N(0,1)$, the transfer function $F(z)$
\[F(z)=-(\b^2-1)\frac{\b^{-1}z^{-1}}{1-\b z^{-1}},\]
which has one pole at $\b=\sqrt{1+P} > 1$, achieves the MMSE capacity 
\[\CE(P)=C(P)=\half \log(1+P)\]
under power constraint $P$.
\end{cor}
\begin{IEEEproof}
Note that for the AWGN channel, $S_Z(e^{j\omega})=1$, and by Corollary~\ref{ARMAProp} we have
\[\CE(P)=C(P)=\half \log\left(1+P\right).\]
Moreover, for $\b=\sqrt{1+P}$, the transfer function $F(z)$ has instability $I=\log(\b)$, hence 
\[\CE(P)=C(P)=I.\]
It is left to show that $F(z)$ also satisfies the power constraint~\eqref{pc2}, i.e.,
\begin{align}\label{pcalias}
\frac{1}{2\pi}\int_{-\pi}^{\pi} |B(e^{j\omega})|^2 S_Z(e^{j\omega}) d\omega \leq P.
\end{align}
where 
\[B(z)=\frac{F(z)}{1-F(z)}=-(\b^2-1)\frac{\b^{-1}z^{-1}}{1-\b^{-1}z^{-1}}.\]
Note that $B(z)$ is same as the optimal filter given in~\cite[Lemma 4.1]{YH} for the AWGN channel and similarly for $\beta > 1$ we have
\begin{align*}
\frac{1}{2\pi}\int_{-\pi}^{\pi} \left|\frac{1}{1-\b^{-1}e^{-j\omega}}\right|^2 d\omega &=\frac{1}{2\pi} \int_{-\pi}^{\pi} \left|\sum_{k=0}^\infty \b^{-k}e^{-jk\omega} \right|^2 d\omega \\
&= \sum_{k=0}^\infty \beta^{-k} \\
&=\frac{1}{1-\b^{-2}}.
\end{align*}
Hence,
\begin{align}
\frac{1}{2\pi}&\int_{-\pi}^{\pi} \left|\frac{F(e^{j\omega})}{1-F(\jo)}\right|^2 S_Z(e^{j\omega}) d\omega \nn\\
%\frac{1}{2\pi}&\int_{-\pi}^{\pi} |B(e^{j\omega})|^2 S_Z(e^{j\omega}) d\omega \nn \\
&=(\b^2-1)^2\frac{\b^{-2}}{1-\b^{-2}}\nn\\
&=(\b^2-1) \nn \\
&=P
\end{align}
and $F(z)$ satisfy the power constraint. 
\end{IEEEproof}

Finally, we show that the filter in Corollary~\ref{exAWGN} corresponds to the Schalkwijk and Kailath scheme~\cite{SchKai,Sch}. The transfer function $F(z)=\frac{X(z)}{Y(z)}$ in Corollary~\ref{exAWGN} can be represented in the time domain as follows,

\begin{align}
X_i-\b X_{i-1}&= -\frac{\b^2-1}{\b} Y_{i-1}\nn
\end{align}
or
\begin{align}\label{time}
X_i&=\b\left(X_{i-1} -\frac{\b^2-1}{\b^2} Y_{i-1}\right).
\end{align}
As it was shown earlier, for the transfer function $F(z)$ in Corollary~\ref{exAWGN} we have $\E(X_i^2)=P$. Hence, $\E(Y_i^2)=1+P$ and 
\begin{align}
\frac{\b^2-1}{\b^2} &=\frac{P}{1+P} \nn \\
&= \frac{\E(X_{i-1}Y_{i-1})}{\E(Y^2_{i-1})} \label{LMMSE}.
\end{align}
Therefore, \eqref{time} can be written as
\begin{align}
X_i&=\b\left(X_{i-1} - \frac{\E(X_{i-1}Y_{i-1})}{\E(Y^2_{i-1})}  Y_{i-1}\right), \nn
\end{align}
which is the recursive representation of the Schalkwijk and Kailath encoder~\cite{YH--Lecture}.

\section{Gaussian Multiple Access Channel}\label{MAC}

           \begin{figure*}[htbp]
   \centering
   \begin{picture}(300,280)(0,0)

     \put(0,120){   
    \psfrag{m1}[b]{\hspace{-1em}\small  $ M_1$}
\psfrag{m2}[b]{\hspace{-1em}\small  $M_j$}
\psfrag{m3}[b]{\hspace{-1em}\small  $M_N$}
\psfrag{e}[b]{\hspace{5em} \small  $\hat{M}_1,\ldots,\hat{M}_N$}
\psfrag{en1}{\hspace{1.4em} { \small Encoder $1$}}
\psfrag{f1}{\hspace{-.4em} { \small $X_{1i}=f_{1i}(M_1,Y^{i-1})$}}
\psfrag{enj}{\hspace{1.4em} { \small Encoder $j$}}
\psfrag{fj}{\hspace{-.5em} { \small $X_{ji}=f_{ji}(M_j,Y^{i-1})$}}
\psfrag{enN}{\hspace{1.4em} { \small Encoder $N$}}
\psfrag{fN}{\hspace{-.85em} { \small $X_{Ni}=f_{Ni}(M_N,Y^{i-1})$}}

\psfrag{de}{\hspace{-.8em} { \small Decoder }}
\psfrag{x1}{\hspace{-.5em} $X_{1i}$}
\psfrag{x2}{\hspace{-.5em} $X_{ji}$}
\psfrag{x3}{\hspace{-.7em} $X_{Ni}$}
\psfrag{y}{\hspace{1em}\small  $Y_i$}
\psfrag{z}{\hspace{-2em} { \small $Z_i \sim \N(0,1)$}}
%\psfrag{F}{\hspace{-1.5em} $Y^{i-1}$}        
\psfrag{d}{\hspace{.5em} $\vdots$}           
  
           \includegraphics[width=4.2in]{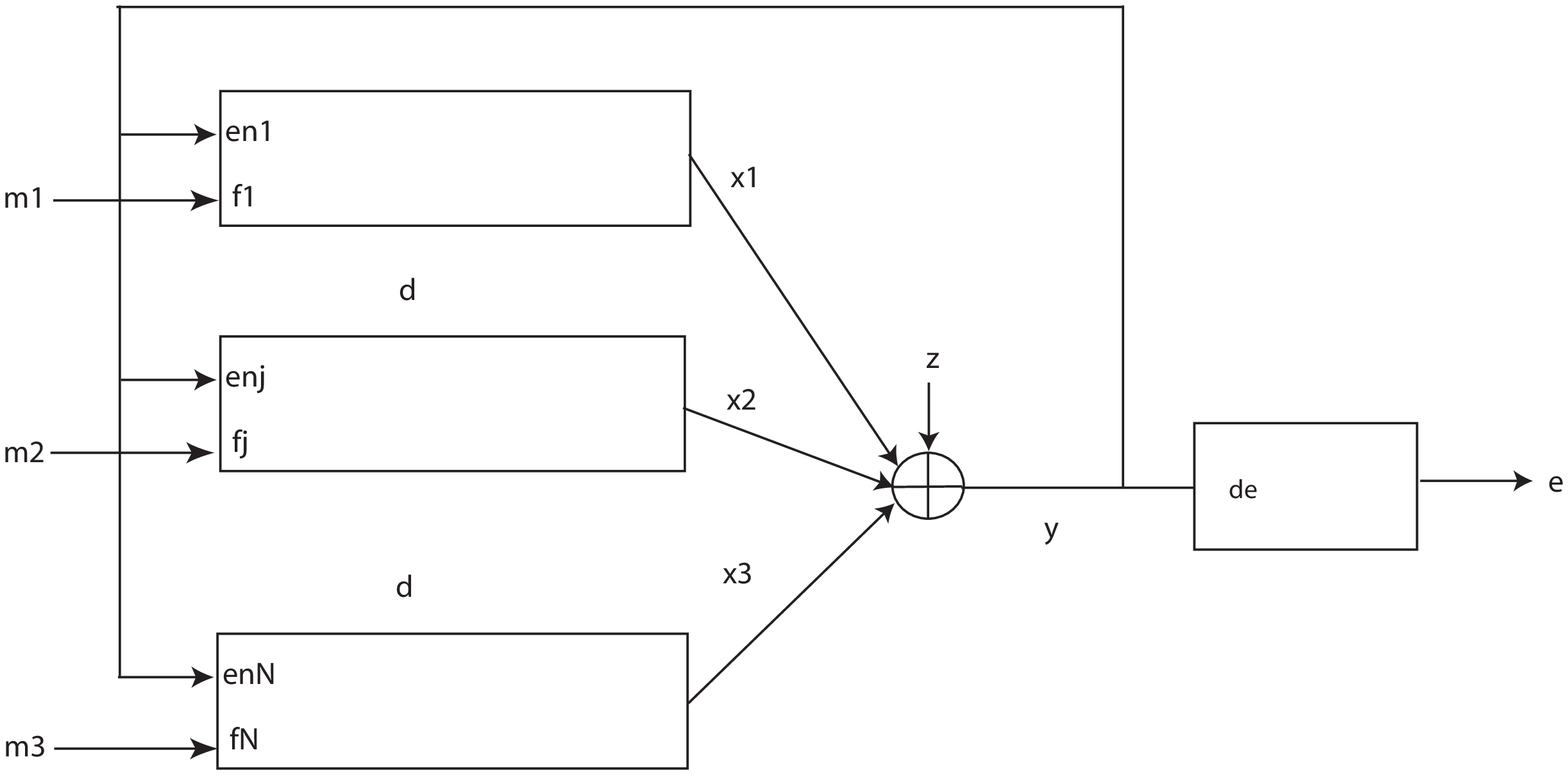}    }       
               
   \put(-100,100){\scriptsize (a) The AWGN-MAC with feedback  }              
                       
     \put(15,0){   
   \psfrag{f}[b]{ \hspace{3.5em}\small $\Sv_{i}=A\Sv_{i-1}+BY_{i-1}$}
\psfrag{s}[b]{ \hspace{3.5em} System}
\psfrag{i}[b]{\hspace{-6em}\small  $\Sv_0=(S_{10},\ldots,S_{N0})$}
\psfrag{y}[b]{\hspace{3em} \small $Y_i$}
\psfrag{u}[b]{\hspace{-2em} \small $X_i$}
\psfrag{z}[b]{\hspace{.3em} \small  $Z_i\sim \N(0,1)$}
\psfrag{q}{\hspace{-.5em} { \small Controller}}
\psfrag{ch}[b]{\hspace{7em} \small Channel}

\psfrag{e}[b]{ \hspace{3.5em} \small  $X_i=\pi_i(\Sv_{i})$}

   \includegraphics[width=2.8in]{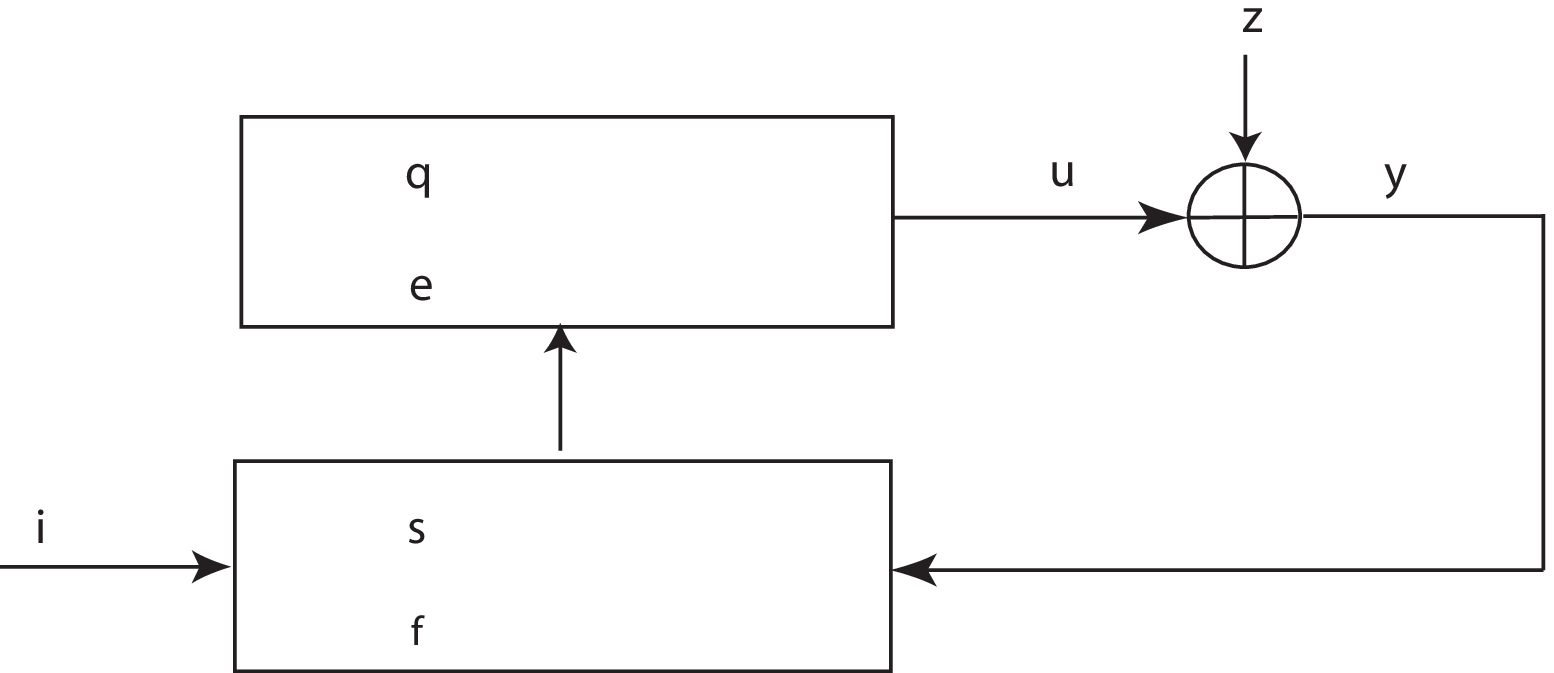} }

\put(-100,-10){\scriptsize (b) Control over AWGN channel  }  
                   
   \end{picture}   
   \vspace{.25in}
   \caption{The set of $N$ encoders for feedback communication over the AWGN-MAC depicted in (a) can be represented as a dynamical system depicted in (b).} 
  \label{fig:GMAC}
\end{figure*}

In this section, we extend the control representation for the communication over point-to-point channels to the feedback communication of $N$ senders and one receiver over the AWGN-MAC depicted in Fig.~\ref{fig:GMAC}a. Each sender $j \in \{1, \ldots,  N\}$ wishes to reliably transmit a message $M_j \in \Mc_j:=(0,1)$ in the unit interval to the receiver. At each time~$i$, the output of the channel is
\begin{align*}
Y_i=\sum_{j=1}^N X_{ji} + Z_i
\end{align*}
where $X_{ji} \in \Real$ is the transmitted symbol by sender~$j$ at time~$i$, $Y_i \in \Real$ is the output of the channel, and $\{Z_i\}$ is a discrete-time zero-mean white Gaussian noise process with unit average power, i.e., $\E(Z^2_i)=1$, independent of $M_1,\ldots, M_N$. We assume that output symbols are causally fed back to the sender and the transmitted symbol $X_{ji}$ for sender~$j$ at time~$i$ can depend on both the message $M_j$ and the previous channel output sequence $Y^{i-1}:= \{\Yi\} $. We define a $n$-code for the AWGN-MAC as
\begin{enumerate}
\item $N$ encoders: each encoder $j$, $j=1,\ldots,N,$ is a set of encoding maps $f_{ji}:  (\Mc_j,\Yc^{i-1}) \to \Xc_j$, $i=1,\ldots,n$, also known to the receiver, such that 
\begin{align}
X_{ji}=f_{ji}(M_j,Y^{i-1}) 
\end{align} 
and 
\item a decoder: a decoding map $\phi: \Yc^n \to \Mc_1 \times \ldots \times \Mc_N$ which determines the estimates of the messages $\Mh_1,\ldots, \Mh_N$ given $Y^n$
\begin{align}
(\Mh_1,\ldots, \Mh_N)=\phi(Y^n).
\end{align}
\end{enumerate}
We assume that the message vector $\Mv:=(M_1,\ldots, M_N) \sim \U\left((0,1)^N\right)$ is uniformly distributed in a $N$ dimensional unit box and the performance measure is the set of mean square errors, 
\[\Dn_j:=\E\left((M_j-\Mh_j(Y^n))^2\right) \quad j=1,\ldots, N.\] 
The set of {\em MSE exponents} $(E_1,\ldots,E_N)$ is called achievable if there exists a sequence of $n$-codes such that for $j=1,\ldots,N,$
\begin{align}
 \liminf_{n \to \infty} -\frac{1}{2n}\log(\Dn_j) &\geq E_j. \nn
 \end{align}
We say that a sequence of $n$-codes has {\em asymptotic powers} $(\Pb_1,\ldots,\Pb_N)$ if  
\begin{align}
  \lim_{n \to \infty} \E(X^2_{jn}) =  \Pb_j, \quad j\in \{1,\ldots,N\}. \nn
 \end{align}

Similar to the point-to-point case, the information-theoretic achievable rates~\cite{Cover--Thomas2006} are defined based on a sequence of massage sets whose size depend on the block length~$n$. Consider the discrete message sets $\Mc_{jn}:=\{1,2,\ldots,2^{nR_j}\}, j=1,\ldots,N$. For a $n$-code with messages $M_{jn} \sim \U(\Mc_n), j=1,\ldots,N$ uniformly distributed over the set $\Mc_{jn}$, let the probability of error be 
\[\pen :=\P\{(M_1,\ldots,M_N)\neq (\Mh_1,\ldots,\Mh_N)\}. \]
The set of rates $(R_1, \ldots, R_N)$ is called achievable  under block power constraints $(P_1,\ldots,P_N)$ if there exists a sequence of $n$-codes with messages $M_j \sim \U(\Mc_{j,n})$, such that for $i \in \{1,\ldots,n\}$, and $j\in \{1,\ldots,N\}$,
\begin{align}\label{avP}
\sum_{i=1}^n \E(X^2_{ji}) \leq nP_j 
\end{align}
and $P_e^{(n)} \to 0$ as $n \to \infty$. We refer to $R=\sum_{j=1}^N R_j$ as the {\em sum rate} of a given code. 

The following lemma presents the connection between the achievable MSE exponents and the  achievable rates.
\begin{lem}\label{lemMSE}
Let MSE exponents $(E_1,\ldots,E_N)$ with asymptotic powers $(\Pb_1,\ldots,\Pb_N)$ be achievable and 
\begin{align*}
R_j< E_j, \ \Pb_j < P_j
\end{align*}
for $j=1,\ldots,N$. Then the rate-tuple $(R_1, \ldots, R_N)$ is also achievable and the block power constraints $(P_1,\ldots,P_N)$ are asymptotically satisfied.  
\end{lem}
\begin{IEEEproof} 
Let the MSE exponents $(E_1,\ldots,E_N)$ be achievable. By a similar argument as in Lemma~\ref{thmMSE} we can see that the rates $(R_1, \ldots, R_N)$ are also achievable if
\begin{align*}
R_j< E_j, \quad j=1,\ldots,N.
\end{align*}
Moreover, if $\Pb_j=\lim_{n \to \infty} \E(X^2_{jn}),$ $j=1,\ldots,N$ 
and $\Pb_j < P_j$, then using the Ces\'{a}ro sum we can see that the constraint 
\begin{align*}
\frac{1}{n}\sum_{i=1}^n \E(X^2_{ji}) \leq P_j , \quad j\in \{1,\ldots,N\}
\end{align*}
is satisfied for sufficiently large $n$. 
\end{IEEEproof}

\subsection{LQG Approach for the AWGN-MAC}
In the following we show that codes can be designed for the AWGN-MAC with feedback, based on the LTI systems controlled over a point-to-point AWGN channel depicted in~Fig.~\ref{fig:GMAC}b. We refer to these codes as LTI codes for the AWGN-MAC.

Let the state $\Sv_i \in \Complex^N$ be a complex vector of length $N$, and the system dynamics be
\begin{align}\label{systemmac}
\Sv_{i}&=A\Sv_{i-1}+BY_{i-1} \quad  i=1,\ldots,n
\end{align} 
where
\begin{align}\label{Aform}
A =  \left( \begin{array}{ccccc}
             \beta_1 \omega_1  &   0 &  0 & \dots & 0 \\
             0 & \beta_2 \omega_2 & 0 & \ldots  & 0 \\
             \vdots & \vdots & \vdots & \ddots & \vdots \\
             0 &  0 & 0 & \ldots  &  \beta_N \omega_{N} \\
        \end{array} \right), \  B=  \left( \begin{array}{c}
             1 \\
             1 \\
             \vdots \\
             1 \\
        \end{array} \right)
  \end{align}
such that $\beta_j > 1$ are real and $|\omega_j|=1$ are distinct points on the unit circle of the complex plane. As it is clear from the system dynamics given in~\eqref{systemmac}, the dependence of the state on the channel outputs is causal .
 
The controller is assumed to control each mode~$j$ separately and transmit a combined control signal (scalar) $X_i$ as follows. It observes the current state $\Sv_{i}$ and picks a vector
\begin{align}\label{actionVec}
\Xv_i:=(X_{1i}, \ldots,X_{Ni})
\end{align}
where 
\begin{align}\label{controlj}
X_{ji}=\pi_{ji}(\Sv_{i}(j)), \quad j=1,\ldots,N
\end{align}
can depend only on $\Sv_{i}(j)$, and then transmits the control signal  
\begin{align}\label{controlsep}
X_i=\sum_{j=1}^N X_{ji}= \sum_{j=1}^N \pi_{ji}(\Sv_{i}(j)).
\end{align} 
For the purpose of analysis we assume complex transmissions over a complex channel, 
\begin{align}\label{channel}
Y_i&=X_i +Z_i, \quad Z_i  \sim \Cc\Nc(0, 1) \ \ \tx{i.i.d.} %\ \mbox{i.i.d.} 
\end{align}
where $X_i,Y_i\in \Complex$ are input and output of the channel, respectively, and  $\{Z_i\}$ is a discrete-time zero-mean white complex Gaussian noise process with identity covariance. Note that one transmission over this complex channel can be viewed as two transmissions over the real channel. Hence, the achievable rates per each complex dimension are also achievable over the real channel.

We assume the complex messages $M_j \in \Complex$, $j=1,\ldots,N,$ in the AWGN-MAC are uniform over $(0,1)\times(0,1)$ and we set the initial state $\Sv_0 \in \Complex^N$ of the system as
\[\Sv_0=\Mv=(M_1,\ldots,M_N).\]
Given the system \eqref{systemmac} and the controller \eqref{controlsep}, we derive the LTI $n$-code for the AWGN-MAC as follows. 
\begin{defn}\label{LTIcode}
An LTI $n$-code for the AWGN-MAC based on LTI system~\eqref{systemmac} and~\eqref{controlsep} consists of
\begin{enumerate}
\item Encoding mappings: The encoder $j$ recursively forms
\begin{align}
\Sv_{i}(j)= \beta_1 \omega_1 \Sv_{i-1}(j)+ Y_{i-1},  \  i=1,\ldots,n \label{encoder}
\end{align}
and at time $i$ transmits $X_{ji}=\pi_{ji}(\Sv_{i}(j))$.
%Then we have
%\[\sum_{j=1}^N x_{ji}=-C\sv_i=X_i\] 
%where the last equality follows from (\ref{control}). 
\item Decoding: At the end of the block, the decoder forms the estimate vector $\hat{\Mv}_n:=-A^{-n}\Shv_n$
where 
\begin{align}
\Shv_{i}&=A\Shv_{i-1}+BY_{i-1}, \quad \Shv_0=0, Y_0=0\label{decoder}
\end{align}
and picks $\hat{\Mv}_n(j)$ as the estimate of the message ${M}_j$.  %which is a linear function of $Y_1,\ldots, Y_{n-1}$.
\end{enumerate}
\end{defn}
Note that to design codes for the MAC it is necessary that the matrix~$A$ is diagonal as in~\eqref{Aform} and that the control signal is of the form~\eqref{controlsep} since the encoders in the MAC are decentralized and do not have access to each other's messages. In the code described above, the dynamics of the $j$-th mode of the system $\Sv_i(j)$ represents the information based on which the encoder~$j$ picks the transmitted signal $X_{ji}$ at time $i$. 
We say that the controller stabilizes the system if 
\[\limsup_{n \to \infty} \E(|\Sv_n(j)|^2) < \infty, \quad j=1,\ldots,N.\]

\begin{lem}\label{1-1}
%Let $A,B,C$ be same as (\ref{Aform}), (\ref{controlapplied}), (\ref{control}). 
If the controller of the form~\eqref{controlsep} stabilizes the linear system in~\eqref{systemmac},  then the corresponding sequence of $n$-codes for the AWGN-MAC with feedback described in~\eqref{encoder} and~\eqref{decoder} achieves MSE exponents 
\[E_j= \log(\beta_j), \quad j=1,\ldots,N.\]
\end{lem}
\IEEEproof See Appendix~\ref{applem1-1}.

\begin{remark}
Note that the set $\{\beta_j\}$ depends only on the system\footnote{From now on,  the word {\em system} refers to the the one given in~\eqref{systemmac}.} and not on the controller. This means that any stabilizing controller for the system~\eqref{systemmac} can be used to design a code for the AWGN-MAC, which achieves the same set of MSE exponents $(\log(\b_1), \ldots, \log(\b_N))$ or, according to Lemma~\ref{thmMSE}, the same set of rates. 
\end{remark}

In the following, we assume that the system is fixed as in~\eqref{systemmac} and we find the stationary controller which minimizes the asymptotic power 
\begin{align}\label{sumpower}
\Pb:=& \lim_{n\to \infty} \frac{1}{n} \sum_{i=1}^n \E(|X_i|^2).
%=&\sum_{j,k=1}^N \Kb_{jk}
%=& \lim_{n\to \infty} \frac{1}{n} \sum_{i=1}^n \E(|\sum_{j=1}^N X_{ji}|^2)\\
%=& \lim_{n\to \infty} \frac{1}{n} \sum_{i=1}^n \E(|\sum_{j=1}^N X_{ji}|^2)\\
\end{align}
As we show below, the code for the AWGN-MAC which is based on this optimal controller is also optimal in terms of sum rate among the linear codes for the AWGN-MAC under equal power constraints. 

For a controller of the form~\eqref{controlsep}, let $K_{\Xv_i}:=\cov(\Xv_i)$ be the covariance of $\Xv_i$ given in~\eqref{actionVec}. We say that $\Kb_{\Xv}$ is  the {\em asymptotic covariance} if
\[\Kb_{\Xv}:= \lim_{n\to \infty} K_{\Xv_n}.\]
Note that by~\eqref{controlsep}, the asymptotic power of the controller can be written as $\Pb=\sum_{j,k=1}^N (\Kb_{\Xv})_{jk}$. Hence, $\Pb$ represents the asymptotic combined power of all the senders in the corresponding code for the AWGN-MAC, which also captures the correlation between the transmitted signals. Whereas, the asymptotic power of each sender in the AWGN-MAC is determined by $\Pb_j=(\Kb_{\Xv})_{jj}$, $j=1,\ldots,N$. 

\begin{lem}\label{lemLQG}
The optimal controller of the form~\eqref{controlsep} which minimizes the cost~\eqref{sumpower} is stationary and linear, i.e., 
\[X_i=-C\Sv_i,\]
where
\begin{align}\label{optC}
C= (B'GB+1)^{-1}B'GA
\end{align}
and $G$ is the solution to the following discrete algebraic Riccati equation (DARE)
\begin{align}\label{riccati}
G=A'GA-A'GB(B'GB+1)^{-1}B'GA.
\end{align}
%or equivalently the following discrete algebraic Lyapunov equation (DALE)
%\begin{align}\label{Lyapunov2}
%G=(A-BC)'G(A-BC)+C'C.
%\end{align}
\end{lem}
\begin{IEEEproof}
Note that the optimal power $\Pb^*$ for the controllers of the form~\eqref{controlsep} can be lower bounded as 
\[\Pb^* \geq \Pb_{\min}\]
where $\Pb_{\min}$ is the optimal cost considering a more general controller of the form 
\begin{align}\label{generalcont}
X_{i}=\pi_{i}(\Sv_i)
\end{align}
which picks the scalar action $X_i$ based on the complete state $\Sv_i$ and is not necessarily separated for different modes as in~\eqref{controlsep}. Considering the general control~\eqref{generalcont} we have a linear Gaussian quadratic control problem. From the theory of LQG~\cite{LinEst} we know that the optimal control is linear, i.e., 
\begin{align}\label{solLQG}
X_i=-C_i\Sv_i, \quad C_i=[c_{1i} c_{2i}  \ldots  c_{Ni}].  
\end{align}
and the stationary linear control $C=\lim_{i \to \infty} C_i$, which minimizes the asymptotic cost~\eqref{sumpower}, is given by $C$ in~\eqref{optC}.
%\begin{align}\label{riccati}
%G=A'GA-A'GB(B'GB+1)^{-1}B'GA
%\end{align}
%or equivalently the following discrete algebraic Lyapunov equation (DALE)
%\begin{align}\label{Lyapunov2}
%G=(A-BC)'G(A-BC)+C'C.
%\end{align}
Since the solution to this LQG problem given in~\eqref{solLQG} is of the form~\eqref{controlsep} we  conclude that
\[\Pb^* = \Pb_{\min}\] 
and the optimal controller of the form~\eqref{controlsep} is same as the optimal controller of the LQG problem. 
\end{IEEEproof}

By Lemma~\ref{lemLQG}, the optimal controller of the form~\eqref{controlsep} is linear. The following theorem provides rate and power analysis for the AWGN-MAC code based on a general stationary linear control of the form
\begin{align}\label{lincontrol}
X_i=-C\Sv_i, \quad C=[c_{1} c_{2}  \ldots  c_{N}].  
\end{align}
\begin{thm}\label{thmlin}
%Let $A,B,C$ be same as (\ref{Aform}), (\ref{controlapplied}), (\ref{control}). 
Consider the stationary linear controller~\eqref{lincontrol} for the system~\eqref{systemmac}. If $A-BC$ is stable, i.e., the eigenvalues of $A-BC $ are inside the unit circle, then the corresponding code for AWGN-MAC designed based on this linear control achieves the rates
\[R_j < E_j=\log(\beta_j), \quad j=1,\ldots,N\] 
with asymptotic powers 
\begin{align}\label{powerj}
\Pb_j=c_j^2\Kb_{jj}, \quad j=1,\ldots,N
\end{align}
where $\Kb$ is the unique solution of the following discrete algebraic Lyapunov equation (DALE)
\begin{align}\label{Lyapunov}
\Kb =  (A-BC)\Kb(A-BC)'+Q,  \quad Q=BB'. 
\end{align}
\end{thm}
\begin{IEEEproof}
First note that by Lemma~\ref{1-1}, MSE exponents $E_j=\log(\b_j)$, $j=1,\ldots,N,$ are achievable and hence, by Lemma~\ref{lemMSE}, any rate $R_j < E_j$, $j=1,\ldots,N$ is achievable in each complex dimension.

To find the asymptotic powers, note that given a stationary linear controller of the form~\eqref{lincontrol}, the closed loop system given~\eqref{systemmac} can be written as  
\begin{align}\label{linclosedloop}
\Sv_{i}&=(A-BC)\Sv_{i-1}+BZ_{i-1}.
\end{align}
Consider the code corresponding to the linear control~\eqref{lincontrol}, where the encoder $j$ transmits
\[X_{ji}=c_j\Sv_i(j)\]
at time~$i$ and let $K_i:=\mbox{Cov}(\Sv_i)$. Then, the asymptotic power $\Pb_j$ is
\begin{align}
\Pb_j&=\lim_{n \to \infty} \frac{1}{n} \sum_{i=1}^n \E(X^2_{ji}) =\lim_{n \to \infty}  \frac{1}{n} \sum_{i=1}^n   |c_j|^2 (K_i)_{jj}. \label{asympower}
\end{align}
From the closed loop dynamic given in~\eqref{linclosedloop} we have the following Lyapunov recursion,\begin{align*}
 K_{n} &=  (A-BC)K_{n-1}(A-BC)' +Q,   \quad  Q=BB'. 
%   Q=BB'&= \left( \begin{array}{ccccc}
 %            1  &   \ldots&  1  \\
 %            1   & \ldots  & 1\\
%             \vdots & \ddots & \vdots \\
 %            1  &  \ldots  &  1 \\
  %      \end{array} \right)
\end{align*} 
If $A-BC$ is stable, we know~\cite{LinEst} that
 \begin{align}\label{limcov}
 \lim_{n \to \infty} K_n = \Kb \succ 0,
 \end{align} 
 where $\Kb$ is the unique positive-definite solution of the corresponding discrete algebraic Lyapunov equation (DALE)
 \[\Kb =  (A-BC)\Kb(A-BC)' +Q.\]
Therefore, by~\eqref{limcov} and the Ces\'{a}ro mean theorem, the asymptotic power in~\eqref{asympower} becomes
$\Pb_j= |c_j|^2 \Kb_{jj}.$
\end{IEEEproof}
%%%%%%%%%%%%%%%%%%%%%% %%%%%%%%%%%%%%%%%%%%%%

%%%%%%%%%%%%%%%%%%%%%%%%%%%%%%%%%%%%%%%%
\subsection{Kramer Code vs. Optimal LTI Code}
Now consider the symmetric system 
\begin{align}
\omega_j&=e^{\frac{2\pi \sqrt{-1}}{N}(j-1)} \nn\\
\beta_j&=\beta \quad j=1,\ldots,N. \label{symmetricA}
 \end{align}
and let the {\em optimal LTI code} for the AWGN-MAC be the code based on the optimal control given in Lemma~\ref{lemLQG}. The following lemma characterizes the matrix $G$ for the symmetric choice of matrix~$A$ given in~\eqref{symmetricA}. 
\begin{lem}\label{symlem}
Let $\b_j=\b$ and $\omega_j= e^{2\pi \sqrt{-1} \frac{(j-1)}{N}}$. The unique positive-definite solution $G \succ 0 $ to the DARE~\eqref{riccati} is circulant with real eigenvalues satisfying $\lambda_i=\frac{1}{\b^2} \lambda_{i-1}$ for $i=2,\ldots, N$. The largest eigenvalue $\l_1$ satisfies
% is equal to $\sigma=\sigma_j=\sum_{k=1}^N \Kb_{jk}, \forall j$ and satisfies
\begin{align}
 1+N\lambda_1 &= \b^{2N}  \label{symsum} \\
\Big(1+\lambda_1\big(N-\frac{\lambda_1}{G_{11}}\big)\Big) &= \b^{2(N-1)}. \label{symother}
\end{align}
%where $P=\Kt_{mm}$ is the main diagonal entries of $\Kt$.
\end{lem}
\IEEEproof See Appendix~\ref{appsymlem}.

The following theorem shows that under equal power constraints for all the senders, the optimal LTI code for the system~\eqref{symmetricA} corresponds to the linear code proposed by Kramer~\cite{KramerFeedback}, which is known to be sum rate optimal within the class of generalized feedback codes~\cite{Ehsan}. 
\begin{thm}\label{Fourier}
Let $R(P)$ be the sum rate achieved by Kramer's code with equal power $P$. There exists a $\beta > 1$ such that the optimal LTI code based on the system~\eqref{symmetricA} achieves the sum rate $R(P)$ and satisfies the equal power constraint~$P$ asymptotically.
\end{thm}
\begin{IEEEproof}
%Let the for AWGN-MAC be the code based on the optimal (linear) controller given by Lemma~\ref{lemLQG} and the symmetric system~\eqref{symmetricA} determined by parameter $\b$. 
By Theorem~\ref{thmlin}, the optimal LTI code based on the symmetric system~\eqref{symmetricA} and the optimal controller in~\eqref{optC} achieves the symmetric MSE exponent $E$, i.e, $E_j=E=\log(\beta), \ j=1,\ldots,N$ with asymptotic powers $\Pb_j=c_j^2\Kb_{jj},\ j=1,\ldots,N$ where $\Kb$ is given in~\eqref{Lyapunov}. Moreover, by Lemma~\ref{lemMSE}, rates $R_j < \log(\beta_j)$, $j=1,\ldots,N$ are achievable and therefore the sum rate
\begin{align}\label{LQGsumrate}
R < N\log(\b)
\end{align} 
is achievable.

On the other hand, let $R(P)$ denote the sum rate achievable by Kramer's code under symmetric per-symbol power $P$. From~\cite{KramerFeedback} we have
\begin{align}
R(P) < \half \log(1+NP\phi(P) ) \label{cl}
\end{align}
where $\phi(P)\in \Real$ is the unique solution in the interval $[1,N]$ to 
\begin{align}\label{phiofP}
(1+NP\phi)^{N-1} = \left(1+P\phi(N-\phi)\right)^N .
\end{align}
Therefore, by picking
\begin{align}\label{sumratebeta}
R(P) < N\log(\beta) <   \log(1+NP\phi(P) )
\end{align}
the LTI code can achieve sum rate equal to $R(P)$. 

It remains to show that the asymptotic powers in the optimal LTI code satisfy
\begin{align}\label{toshow}
\Pb_j < P, \quad j=1,\ldots,N.
\end{align}
%Then, by~\eqref{twoapp} we can conclude that LQG code can achieve sum rate $R(P)$ under symmetric block power constraint $P$. 
Towards this end, we first show that $\Pb_j=G_{jj}$, where $G$ is the solution to~\eqref{riccati}. Note that the DARE given in~\eqref{riccati} can be equivalently written as the following discrete algebraic Lyapunov equation (DALE)
\begin{align}\label{Lyapunov2}
G=(A-BC)'G(A-BC)+C'C
\end{align}
and compared with (\ref{Lyapunov}), the diagonal elements of $G$ and $K$ can be related as follows
\begin{align}\label{GKb}
G_{jj}=c_j^2\Kb_{jj}.
\end{align}
Besides, from~\eqref{powerj} we know $\Pb_j=c_j^2\Kb_{jj}$. Hence, we conclude that the diagonal elements of $G$ represents the asymptotic powers in the optimal LTI code,
\begin{align}\label{PjGj}
\Pb_j=G_{jj}.
\end{align}
Now we will show that if $\beta$ satisfies $\eqref{sumratebeta}$, then $G_{jj} < P$ for the corresponding $G$. First, note that from (\ref{symsum}) and (\ref{symother}) we have 
\begin{align}\label{last}
\Big(1+N\l_1\Big)^{N-1}= \Big(1+\l_1(N-\frac{\l_1}{G_{11}})\Big)^N.
\end{align}
Comparing with (\ref{phiofP}) and noting that $G$ is circulant we get 
\begin{align}\label{lambdaG}
\l_1=G_{jj}\phi(G_{jj}), \quad j=1,\ldots,N.
\end{align}
Now if~\eqref{sumratebeta} holds, from (\ref{symsum}) we know
\begin{align*}
\half \log (1+N\l_1) = N\log(\beta) < \half \log (1+NP\phi(P)) 
\end{align*}
and hence 
\begin{align}\label{lambdaless}
\l_1< P\phi(P).
\end{align} 
From~\eqref{lambdaG} and~\eqref{lambdaless} we have $G_{jj}\phi(G_{jj}) < P\phi(P)$, and $G_{jj} < P$ follows from monotonicity of $\phi(P)$ in $P$. Combined with~\eqref{PjGj}, then~\eqref{toshow} follows. 
\end{IEEEproof}

%%%%%%%%%%%%%%%%%%Section Conclusion%%%%%%%%%%%%%%%%%%%%%
\section{Conclusion}\label{con}
Communication and control problems have different goals, but  they both deal with information dynamics and in many cases they share similar formulations which can be tackled by tools and techniques developed in both fields. Understanding  the interface between these two theories has become even more important in the last decade, as we have witnessed technological advancements leading to the convergence of computing, communication, and control over networked platforms of  embedded systems.  This paper attempted, \emph{ad rem}, one step in this direction.
% showing that tools from control and estimation can be applied to design codes, as well as to reveal fundamental limits of communication in the presence of a feedback link. 

%In particular, for a general point-to-point channel with feedback, information theoretical tools are utilized to show that the maximum decay rate of the MMSE is upper bounded by the channel capacity. This bound is shown to be tight for a special class of Gaussian channels. Moreover, the capacity of Gaussian channels is shown to be equal to the maximum instability tolerable with limited power, a fundamental limitation in control. As far as design issues, the LQG control theory is exploited to derive codes for the communication problem over the $N$-sender AWGN--MAC with feedback. These codes are the stationary version of the linear codes proposed by Kramer and the performance equivalence of the two is shown explicitly for the symmetric case. 

\appendices

%%%%%%%%%%%%%%%%%%%%%%%%Appendix Proof of Lemma Stochastic
\section{Proof of Lemma~\ref{infostate}}\label{appleminfo}
Note that without loss of generality we can decompose the encoding functions $f_i$ in~\eqref{transmitted} into two steps as follows. First, based on the past output $Y^{i-1}$, the encoder picks a function
\begin{align}\label{functiont}
\ft_i: \Mc \to X.
\end{align}
Then it transmits
\[X_i=\ft_i(M).\]
%It is immediate that the two stage encoding described above includes all the mappings in~\eqref{encoderdef}. 
With this separation, the encoder (controller) can be defined by the set of mappings
\[\tilde{\pi}_i: \Yc^{i-1} \to \{\ft\}\]
where $\{\ft\}$ is the set of all functions given in~\eqref{functiont}. Note that this new encoder $\{\tilde{\pi_i}\}_{i=1}^n$ only observes the past output $Y^{i-1}$ and not the message $M$. Therefore, one can view the feedback communication as a control problem where the state is $M$ and the controller has partial observations $Y^{i-1}$. Based on  standard results in stochastic control~\cite[Chapter 6]{Kumar}, there is no loss of optimality if the mapping~$\tilde{\pi}_i$ is picked only according to the conditional distribution $F_{M|Y^{i-1}}(\cdot|Y^{i-1})$. 
Therefore, to find $X_i$ it is sufficient to have the message $M$ and its posterior distribution $F_{M|Y^{i-1}}(\cdot|Y^{i-1})$.

%%%%%%%%%%%%%%%%%%%%%%%%%%%Appendix Proof of Lemma MSE
\section{Proof of Lemma~\ref{thmMSE}}\label{appthmMSE}
Consider a sequence of $n$-codes for $M \in (0,1)$ such that 
\begin{align}\label{appMSEdef}
E \leq \liminf_{n \to \infty} -\frac{1}{2n}\log(\Dn).
\end{align}
We show that a rate $R < E$ is achievable. Consider
\begin{align}
\P( |M-\Mh| >\half \cdot 2^{-nR}) &\leq 4\cdot2^{2nR}\cdot \Dn \label{Che} \\
& \leq 4\cdot2^{2nR}\cdot 2^{-2n(E-\e_n)} \label{Edefref} \\
&=4\cdot 2^{-2n(E-R-\e_n)} \label{RlessE}
\end{align}
for some $\e_n$ where $\e_n \to 0$ as $n \to \infty$, where \eqref{Che} follows from Chebyshev inequality and \eqref{Edefref} follows from~\eqref{appMSEdef}. If $R < E$, from~\eqref{RlessE} we have
\begin{align}
\P( |M-\Mh| >\half \cdot 2^{-nR}) \to 0 \nn
\end{align}
as $n \to \infty$ and the decoder can pick intervals $\Delta_n$ such that 
\[|\Delta_n|=2^{-nR}\] 
and 
\begin{align}\label{intervaldec}
\P(M \notin \Delta_n) \to 0 \ \tx{as} \ n \to \infty. 
\end{align}
By Lemma II.3 in~\cite{Oferarxiv}, if condition~\eqref{intervaldec} is satisfied then for all~$n$, we can map the message set $\Mc_n=\{1,\ldots, 2^{nR}\}$ into a set of message points $\tilde{\Mc}_n \in (0,1)$ in the unit interval such that the minimum distance between the two message points is $2^{-nR}$. Therefore, using the interval decoder described above which satisfy~\eqref{intervaldec}, $\pen \to 0$ as $n \to \infty$ and rate $R$ is achievable.

To complete the proof, let rate $R$ be achievable by a given sequence of $n$-codes such that $\lim_{n \to \infty} \frac{\pen}{2^{-2nR}} \to 0$. For each~$n$, we divide the unit interval~$(0,1)$ into $2^{-nR}$ equal sub-intervals and map the continuous message $M \in (0,1)$ to the discrete message $M_n \in \Mc_n=\{1,\ldots, 2^{nR}\}$ according to the sub-interval $M$ lies in. To communicate $M$, we send the corresponding $M_n$ using the given $n$-code, and we pick the middle point of the interval corresponding to the decoded message $\Mh_n$ as the estimate of $M \in (0,1)$. The MSE for $M$ can be (loosely) upper bounded by $2^{-2nR}$ if the message $M_n$ is decoded correctly, and by $1$ in case of an error. Hence,
\begin{align}\label{errnoerr}
\Dn &\leq \pen + (1-\pen) 2^{-2nR}
%& \leq \alpha_n 2^{-2nR}
\end{align}
%for some $\alpha_n < \infty$. The last inequality follows 
From~\eqref{errnoerr} and the assumption $\lim_{n \to \infty} \frac{\pen}{2^{-2nR}} \to 0$ we have
\[\liminf_{n \to \infty} -\frac{1}{2n} \Dn \geq R\] 
and the MSE exponent $E=R$ is achievable.
%%%%%%%%%%%%%%%%%%%%%%%%%Appendix Proof of Lemma

%%%%%%%%%%%%%%%%%%%%Proof of Lemma generalupp
\section{Proof of Lemma~\ref{generalupp}}\label{appgeneralupp}
We show that 
\[h(Y^n) \geq h(Z^n)+I(S_0;Y^n) \]
with equality if and only if $X_i$ is a function of $(S_0,Y^{i-1})$ for $i=1,\ldots,n$. 

Consider
{\allowdisplaybreaks 
\begin{align}
h(Y^n)&=h(Y^n|S_0)+I(S_0;Y^n)\nn\\
&=\sum_{i=1}^n h(Y_i|Y^{i-1}, S_0)+I(S_0;Y^n)\nn\\
&\geq \sum_{i=1}^n h(Y_i|Y^{i-1}, X^i, S_0)+I(S_0;Y^n)\label{addV}\\
%&\geq \sum_{i=1}^n h(Y_i|Y^{i-1}, X^{i-1})+I(S_0;Y^n)\label{addX}\\
&=\sum_{i=1}^n h(Y_i, Z_i|Y^{i-1}, X^i, S_0)+I(S_0;Y^n)\label{delX}\\
&=\sum_{i=1}^n h(Z_i|Y^{i-1}, X^i, S_0)+I(S_0;Y^n) \nn\\
&= \sum_{i=1}^n h(Z_i|Y^{i-1}, X^i, S_0,Z^{i-1})+I(S_0;Y^n)\label{AddZ}\\
&=\sum_{i=1}^n h(Z_i|Z^{i-1})+I(S_0;Y^n)\label{delRest}\\
&=h(Z^n)+I(S_0;Y^n) \label{thm4.2} 
\end{align}}
where 
\begin{itemize}
\item the inequality~\eqref{addV} comes from the fact that conditioning reduces entropy, and equality holds iff $Y_i \to (S_0,Y^{i-1}) \to X_i$ form a Markov chain for $i=1,\ldots,n$, and since $Y_i=X_i+Z_i$ this Markov chain holds iff $X_i$ is a deterministic function of $(S_0,Y^{i-1})$,
%\item equality \eqref{addX} follows from the fact
\item the equality \eqref{delX} and \eqref{AddZ} follows from the fact that the channel is invertible, and %\item inequality \eqref{AddZ} follows from the fact that conditioning reduces entropy, and equality holds iff $Z_i$ can be determined from $(X_i,Y_i)$, and
\item the equality~\eqref{delRest} follows since $Z_i \to Z^{i-1} \to (Y^{i-1},X^i,S_0)$ form a Markov chain.
\end{itemize}

%%%%%%%%%%%%%%%%%%%%%%%%%%%Appendix Proof 1%%%%%%%%%%%%%%

%%%%%%%%%%%%%%%%%%%Proof of Lemma 1-1

\section{Proof of Lemma~\ref{1-1}}\label{applem1-1}
The system dynamics given in~(\ref{systemmac}) can be rewritten as
\begin{align}
\Sv_{i}&=A \Sv_{i-1}+B Y_{i-1} \nn\\
  %  &=A^i \Sv_0+A \Shv_{i-1}+B Y_{i-1} \nn \\
    &=A^i \Sv_0+\Shv_{i} \label{SiShi}
\end{align}
where $\Shv_i$ is given in~\eqref{decoder}. Plugging $i=n$ and $\Sv_0=\Mv$ into~\eqref{SiShi} and multiplying both sides by $A^{-n}$ we have 
\begin{align*}
A^{-n}\Sv_n&=\Mv+A^{-n}\Shv_n \\
&=\Mv-\hat{\Mv}_n
\end{align*}
where $\hat{\Mv}=-A^{-n}\Shv_n$ is the estimate of the message (see Definition~\ref{LTIcode}). The covariance matrix of the error vector $\ev_n:= \Mv - \hat{\Mv}_n=A^{-n}\Sv_n$ can be written as 
\[\cov(\ev_n)=A^{-n}K_n {A'}^{-n} \]
where $K_n:=\mbox{Cov}(\Sv_n)$ is the covariance matrix of $\Sv_n$. Therefore, the MSE for the sender $j$ is 
\begin{align}
\Dn_j&=\E\left(|\ev_n(j)|^2\right) = \beta_j^{-2n} (K_n)_{jj}. \label{Ejbeta}
\end{align}
By the assumption of stability $\limsup_{n \to \infty}(K_n)_{jj} < \infty$ and from~\eqref{Ejbeta} we have
\begin{align}
 \liminf_{n \to \infty} -\frac{1}{2n}\log(\Dn_j)= \log(\beta_j). \nn
 \end{align}
Hence, the MSE exponents $E_j=\log(\beta_j)$ for $j=1,\ldots, N$ are achievable. 

%%%%%%%%%%%%%%%%%%%%%%%%%%%%%%%%%%%%%%%%%%%%%%%%%%%%%

%%%%%%%%%%%%%%%%%%%%%%%%Appendix Proof Lemma symlem
\section{Proof of Lemma~\ref{symlem}}\label{appsymlem}
Let $A$ and $B$ be of the form \eqref{Aform} with symmetric parameters $\b_j=\b$ and $\omega_j= e^{2\pi \sqrt{-1} \frac{(j-1)}{N}}$. Then we show that the unique positive-definite solution $G \succ 0 $ 
%of the following discrete algebraic Lyapunov equation (DALE),
%\begin{align}\label{DALEapp}
%G=(A-BC)'G(A-BC)+C'C,
%\end{align}
%or equivalently 
the following discrete algebraic Riccati equation (DARE)
\begin{align}\label{riccatiapp}
G=A'GA-A'GB(B'GB+1)^{-1}B'GA
\end{align}
is circulant with real eigenvalues satisfying $\lambda_i=\frac{1}{\b^2} \lambda_{i-1}$ for $i=2,\ldots, N$, and the largest eigenvalue $\l_1$ satisfies
% is equal to $\sigma=\sigma_j=\sum_{k=1}^N \Kb_{jk}, \forall j$ and satisfies
\begin{align*}
 1+N\lambda_1 &= \b^{2N} \\ %\label{symsum} \\
\Big(1+\lambda_1\big(N-\frac{\lambda_1}{G_{11}}\big)\Big) &= \b^{2(N-1)}. %\label{symother}
\end{align*}
%where $P=\Kt_{mm}$ is the main diagonal entries of $\Kt$.
Note that any circulant matrix can be written as $Q\Lambda Q'$, where $Q$ is the $N$ point DFT matrix with
\begin{align}
Q_{jk}=\frac{1}{\sqrt{N}}e^{-2\pi \sqrt{-1} (j-1)(k-1)/N} \label{Qform}
\end{align}
and $\Lambda=\mbox{diag}([\lambda_1, \ldots, \lambda_N])$ is the matrix with eigenvalues on its diagonal. We show that the circulant matrix $\Gt= Q\Lambda Q'$ with positive $\lambda_j > 0$, such that $\lambda_j=\lambda_{j-1}/\beta^2$ for $j \geq 2$, satisfies the Riccati equation (\ref{riccatiapp}). Plugging $Q\Lambda Q'$ into (\ref{riccatiapp}) and rearranging we get
\begin{align*}
\Lambda&= (Q'AQ) \Lambda (Q'AQ)' - ((Q'AQ) \  \Lambda \  (Q'B))  \\ 
&\hspace{.5in} (1+ B'Q \Lambda Q'B)^{-1} ((Q'AQ) \  \Lambda \ (Q'B))'.
\end{align*}
For this symmetric choice of $A$ we have
\begin{align*}
Q'AQ& = \b \left( \begin{array}{ccccc}
          0 &   1            & 0& \ldots      & 0  \\
          0 &   0              & 1& \ldots        &  0 \\
                \vdots &  \vdots         &  \ddots &\ddots      & \vdots               \\
           0 &  0               & \ldots    &   0 &  1\\
           1&      0            & \ldots     & 0  & 0 \\
                 \end{array} \right),    Q'B &=  \left( \begin{array}{c}
          \sqrt{N}\\
           0\\
           \vdots \\
           0\\
                 \end{array} \right).
     \end{align*}
Hence,
\begin{align*}
(Q'AQ)\Lambda (Q'AQ)' &= \b^2 \left( \begin{array}{ccccc}
          \lambda_2 &   0      & \ldots        & 0  \\
          0  &    \lambda_3              & \ldots        &  0 \\
              \vdots &  \vdots         &  \ddots       & \vdots               \\
           0&      0            & \ldots      & \lambda_1\\
                 \end{array} \right)\\
 (Q'AQ)\Lambda (Q'B) &=  \left( \begin{array}{c}
          0\\
           0\\
           \vdots \\
          \b\lambda_1 \sqrt{N} \\
                 \end{array} \right)
     \end{align*}
and the Riccati equation is transformed into $N$ diagonal equations. The first $N-1$ equations are
\begin{align}
\lambda_j=\beta^2 \lambda_{j+1}, \quad j=1,\ldots, N-1 \label{eigenrecursion}
\end{align}
and the $N$-th equation is
\begin{align}
\lambda_N= \b^2 \l_1- \frac{\b^2\l^2_1 N}{1+N\l_1}. \label{ntheq}
\end{align}
From (\ref{eigenrecursion}) we see that $\lambda_1$ is the largest eigenvalue and $\l_N= \b^{-2(N-1)} \l_1$. Combining with (\ref{ntheq}) we get 
\begin{align}\label{bet2N}
(1+N\l_1)=\b^{2N}.
\end{align}
Hence, $\l_1$ is real and so are $\l_j, j=2,\ldots,N$. 

On the other hand, we consider the diagonal equations of the original DARE in~\eqref{riccatiapp}. First, note that from the form of $Q$ in (\ref{Qform}), $\l_1=\sigma_1$ where
\[\sigma_j:=\sum_{k=1}^N \Gt_{jk}. \quad \] 
Moreover, since $\Gt$ is circulant we know that $\sigma_j=\sigma_1$ and $G_{jj}=G_{11}$ for all~$j=1,\ldots, N$, and $(1+B'\Gt B)=1+N\lambda_1$. Hence, the diagonal equations of \eqref{riccatiapp}, i.e., 
\begin{align}
\Gt_{jj}=\b^2 \Gt_{jj}  - \b^2 \frac{\sigma^2_j }{(1+B'\Gt B)}, \quad j=1,\ldots,N
\end{align}
are equivalent to
\begin{align}
\b^2= \frac{1+N\lambda_1}{1+\l_1\Big(N-\frac{\l_1}{\Gt_{11}}\Big)}
\end{align}
and by $(\ref{bet2N})$ we have
\begin{align}\label{beta2N-1}
\Big(1+\lambda_1\big(N-\frac{\lambda_1}{\Gt_{jj}}\big)\Big) &= \b^{2(N-1)}.
\end{align}
Combining \eqref{eigenrecursion}, \eqref{bet2N}, and \eqref{beta2N-1} completes the proof. 
%%%%%%%%%%%%%%%%%%%%%%%%%%%%%%%%%%%%%%%%%%%%%%%%%%%%%
%%%%%%%%%%%%%%%%%%%%%%%Section Acknowledge %%%%%%%%%%%%%%%%%%%%%
\section*{Acknowledgments}
The authors would like to thank Tara Javidi, Young-Han Kim, Paolo Minero, and Ofer Shayevitz for valuable discussions.

\bibliographystyle{IEEEtran}
\bibliography{bibliography}
\end{document}